\documentclass[11pt,a4paper]{article}
\pdfoutput=1

\usepackage[utf8]{inputenc}
\usepackage{hyperref}
\usepackage[numbers,sort&compress]{natbib}
\usepackage{amsmath,amsthm,booktabs,color,graphicx,url}
\usepackage[expansion=false]{microtype}

\newcommand{\myetal}{et al.\mbox{}}
\newcommand{\mysize}[1]{{\lvert #1 \rvert}}
\newcommand{\mythmcite}[1]{{\normalfont\cite{#1}}}
\newcommand{\myA}{\mathcal{A}}
\newcommand{\myC}{\mathcal{C}}
\newcommand{\myG}{\mathcal{G}}
\newcommand{\myH}{\mathcal{H}}
\newcommand{\eps}{\varepsilon}

\newcommand{\mydoi}[2]{\href{http://dx.doi.org/#1}{#2}}
\newcommand{\myurl}[2]{\href{#1}{#2}}

\newtheorem{theorem}{Theorem}

\newtheorem{corollary}{Corollary}

\theoremstyle{remark}
\newtheorem{remark}{Remark}

\hypersetup{
    colorlinks=false,
    pdfborder={0 0 0},
    pdftitle={Local algorithms in (weakly) coloured graphs},
    pdfauthor={Matti Åstrand, Valentin Polishchuk, Joel Rybicki, Jukka Suomela, Jara Uitto},
}

\begin{document}

\title{Local algorithms in (weakly) coloured graphs}
\author{Matti Åstrand, Valentin Polishchuk, Joel Rybicki, \\ Jukka Suomela, and Jara Uitto}
\date{January 2010}
\maketitle
\begin{abstract}
    A local algorithm is a distributed algorithm that completes after a constant number of synchronous communication rounds. We present local approximation algorithms for the minimum dominating set problem and the maximum matching problem in $2$-coloured and weakly $2$-coloured graphs. In a weakly $2$-coloured graph, both problems admit a local algorithm with the approximation factor ${(\Delta+1)}/2$, where $\Delta$ is the maximum degree of the graph. We also give a matching lower bound proving that there is no local algorithm with a better approximation factor for either of these problems. Furthermore, we show that the stronger assumption of a $2$-colouring does not help in the case of the dominating set problem, but there is a local approximation scheme for the maximum matching problem in $2$-coloured graphs.
\end{abstract}

\section{Introduction}

A \emph{$2$-colouring} of a graph is an assignment of the black or white colour to the nodes so that each black node is adjacent only to white nodes, and vice versa. A \emph{weak $2$-colouring} assigns the colours so that each non-isolated black node is adjacent to at least one white node, and vice versa.

A graph can be $2$-coloured if and only if it is bipartite; a weak $2$-colouring always exists. Given a global view of the graph, it is easy to find a $2$-colouring of a bipartite graph and a weak $2$-colouring of any graph.

In a distributed setting, it is not possible to $2$-colour a bipartite graph without essentially global information of the whole graph. However, Naor and Stockmeyer~\cite{naor95what} showed in 1995 that one can find a weak $2$-colouring with a \emph{constant-time} synchronous distributed algorithm, assuming that the degree of each node is \emph{odd} and bounded by a constant.

Constant-time distributed algorithms are known as \emph{local algorithms} \cite{naor95what,suomela09survey} -- in a local algorithm, the output of each node depends only on its local neighbourhood, and the radius of the neighbourhood does not depend on the number of nodes in the network.

\subsection{Contributions}

We present local approximation algorithms for both $2$-coloured and weakly $2$-coloured graphs. We assume that a colouring is given in the input, i.e., that every node knows its colour. We study exactly how much this additional information helps from the perspective of local approximation algorithms.

We focus on two classical problems -- minimum dominating set and maximum matching. We consider bounded-degree graphs; we assume that there is a known constant $\Delta$ such that the degree of any node is at most~$\Delta$. The results are summarised in Table~\ref{tab:results}. All results are tight: there are matching upper and lower bounds.

\begin{table}
    \centering
    \newcommand{\mysp}{\hspace{1.5em}}
    \begin{tabular}{l@{\mysp}l@{\mysp}l@{\mysp}l}
        \toprule
        Problem & approx. factor & upper bound & lower bound \\
        \midrule
        dominating set \\
        -- no colouring, even $\Delta$ & $\Delta+1$ & trivial & \cite{czygrinow08fast,lenzen08leveraging} \\
        -- no colouring, odd $\Delta$ & $\Delta$  & Theorem~\ref{thm:pos-ds-odd} & \cite{czygrinow08fast,lenzen08leveraging} \\
        -- weak $2$-colouring & $(\Delta+1)/2$ & Theorem~\ref{thm:pos-ds-weak} & Theorem~\ref{thm:neg-ds-strong} \\
        -- $2$-colouring & $(\Delta+1)/2$ & Theorem~\ref{thm:pos-ds-weak} & Theorem~\ref{thm:neg-ds-strong} \\
        \midrule
        matching \\
        -- no colouring & none & --- & \cite{czygrinow08fast} \\
        -- weak $2$-colouring & $(\Delta+1)/2$ & Theorem~\ref{thm:pos-m-weak} & Theorem~\ref{thm:neg-m-weak} \\
        -- $2$-colouring & $1+\eps$ & Theorem~\ref{thm:pos-m-strong} & --- \\
        \bottomrule
    \end{tabular}
    \caption{The best possible approximation factors achievable by a deterministic local algorithm.}\label{tab:results}
\end{table}

In particular, we show that a weak $2$-colouring is as good as a $2$-colouring from the perspective of the local approximability of the dominating set problem. Furthermore, a weak $2$-colouring provides enough symmetry-breaking information so that an approximation of a maximum matching can be found locally. Finally, with a $2$-colouring, the maximum matching can be approximated to within an arbitrary constant.

We also look at a third problem, maximum independent set. There is a trivial local $\Delta$-approximation algorithm for independent set in $2$-coloured graphs: take all white nodes and all isolated black nodes. However, we show that in weakly $2$-coloured graphs, the problem does not admit any local constant-factor approximation algorithm  (Theorem~\ref{thm:neg-is-weak}).

\subsection{Model of distributed computing}

Let $\myG = (V, E)$ be a graph identified with a distributed system. Each node $v \in V$ is a device; there is an edge $\{u,v\} \in E$ if $u$ and $v$ can communicate with each other. To avoid trivialities, we assume that there are no isolated nodes in~$\myG$.

Each node runs the same deterministic algorithm $\myA$. Communication is synchronous: on every time step, all nodes first receive messages from their neighbours, then all nodes perform local computation, and finally all nodes send messages to their neighbours. The algorithm running in the node $v \in V$ knows the degree of $v$. Furthermore, if the node $v$ has a label (such as a colour or a unique identifier), then $\myA$ has access to the label. If the edges are oriented, then $\myA$ knows which edges are outgoing and which edges are incoming. Finally, every node knows the maximum node degree $\Delta$.

The distributed algorithm $\myA$ is a local algorithm if there is a constant $T$ such that the algorithm completes in $T$ synchronous communication rounds, regardless of the input graph $\myG$. The algorithm $\myA$ and the constant $T$ may depend on the degree bound $\Delta$; however, the time $T$ cannot depend on the number of nodes in $\myG$.

The results that we present are essentially oblivious to any other details of the model of distributed computing. All lower bounds (impossibility results) hold even if we use Linial's~\cite{linial92locality} model. We can assume that each node knows the total number of nodes $\mysize{V}$, each node is assigned a unique identifier from the set $\{1, 2, \dotsc, \mysize{V}\}$, local computation is free, and the size of a message is unbounded.

Our upper bounds (algorithms) do not need to exploit any of these assumptions. The nodes do not need to know $\mysize{V}$. Local computations are simple and messages are small; in particular, the size of a message does not depend on $\mysize{V}$. Furthermore, the algorithms do not require unique identifiers. With the exception of Theorem~\ref{thm:pos-ds-odd}, the algorithms only assume that there is a \emph{port numbering}~\cite{angluin80local}: each node imposes an ordering on incident edges.

\section{Prior work}\label{sec:prior}

Randomised local algorithms exist for dominating set~\cite{kuhn05constant-time,kuhn05price,kuhn06price}, matching~\cite{wattenhofer04distributed,hoepman06efficient,nguyen08constant-time}, and independent set~\cite{czygrinow08fast}. However, deterministic local approximation algorithms are scarce. The set of all nodes is a trivial ${(\Delta+1)}$-approximation of a minimum dominating set, and there are local constant-factor approximation algorithms for dominating set in planar graphs~\cite{czygrinow08fast,lenzen08what}. Some positive results are known for matchings in bounded-degree $2$-coloured graphs: local algorithms exist for finding a maximal matching~\cite{hanckowiak98distributed} and a constant-factor approximation of a maximum-weight matching~\cite{floreen09almost-stable}.

To present the earlier negative results on which we build our lower bounds, we need the following definition: a \emph{numbered directed $n$-cycle} $\myC$ is a directed $n$-cycle where each node is assigned a unique identifier from the set $\{1, 2, \dotsc, n\}$. Each node has one incoming and one outgoing edge.

Linial's~\cite{linial92locality} seminal work shows that there is no local algorithm for finding a maximal independent set in $\myC$. Recently, Czygrinow \myetal{}~\cite{czygrinow08fast} and Lenzen and Wattenhofer~\cite{lenzen08leveraging} have extended this result to the approximability of the maximum independent set problem. We include here a proof of an adaptation of the inapproximability result, since the proofs of our lower bounds build directly upon it. We follow Czygrinow \myetal{}'s~\cite{czygrinow08fast} techniques. In the proof, a \emph{$k$-set} is a set with $k$ elements.
\begin{theorem}[\mythmcite{czygrinow08fast,lenzen08leveraging}]\label{thm:neg-is-cycle}
    For any $\alpha \ge 1$ and any local algorithm $\myA$, there exists an integer $n_0$ such that for every $n\ge n_0$ there is a numbered directed $n$-cycle $\myC$ where $\myA$ does not produce an $\alpha$-approximation for maximum independent set.
\end{theorem}
\begin{proof}
    Denote by $T$ the number of synchronous communication rounds that $\myA$ takes. Let $m=\lceil 16 T \alpha\rceil$. By Ramsey's theorem~\cite{ramsey30problem}, there is a finite $N$ with the following property: Let $S$ be a set with at least $N$ elements, and assign an arbitrary label $f(X) \in \{0,1\}$ to each ${(2T+1)}$-set $X \subset S$. Then there is an $m$-set $A \subset S$ and a label $\ell \in \{0,1\}$ such that $f(X) = \ell$ for every ${(2T+1)}$-set $X \subset A$. We say that $A$ is an $\ell$-coloured $m$-subset of $S$.

    Let $n_0=\lceil 8 N \alpha\rceil$, let $n \ge n_0$ and $S = \{1, 2, \dotsc, n\}$. We first assign a label $f(X) \in \{0,1\}$ to each ${(2T+1)}$-set $X \subset S$. Let $X = \{x_1, x_2, \dotsc, x_{2T+1}\}$ with $x_1 < x_2 < \dotso < x_{2T+1}$. Consider a fragment of a numbered directed $n$-cycle with the unique identifiers $x_1, x_2, \dotsc, x_{2T+1}$, in this order; let $f(X) \in \{0,1\}$ be the output of $\myA$ for the node $x_{T+1}$, with $1$ denoting that the node joins the independent set. Observe that the output only depends on the set~$X$.

    Let us next construct a numbered directed $n$-cycle $\myC$ as follows. By the choice of $n_0$, we can find an $\ell_1$-coloured $m$-subset $A_1$ of $S$ for an $\ell_1 \in \{0,1\}$. As $\mysize{S \setminus A_1} \ge N$, we can then find an $\ell_2$-coloured $m$-subset $A_2$ of $S \setminus A_1$ for an $\ell_2 \in \{0,1\}$, etc. Overall we find $p=\lceil{(n-N)}/m\rceil$ disjoint sets $A_1, A_2, \dotsc, A_p$ such that $A_i$ is an $\ell_i$-coloured $m$-subset of $S$. Let $A_i = \{a_i^1, a_i^2, \dotsc, a_i^m\}$ with $a_i^1 < a_i^2 < \dotso < a_i^m$ for each $i$. Let $S \setminus (\bigcup_i A_i) = \{s_1, s_2, \dotsc, s_k\}$. Assign the unique identifiers in $\myC$ in the order
\[
    a_1^1, a_1^2, \dotsc, a_1^m, a_2^1, a_2^2, \dotsc, a_p^m, s_1, s_2, \dotsc, s_k.
\]

    An optimal independent set of $\myC$ contains at least $n/3$ nodes; to prove the theorem, it suffices to show that the algorithm $\myA$ outputs $1$ for at most $n/(4\alpha)$ nodes. To see this, observe that for each $i$ the output of the nodes $a_i^{T+1}, a_i^{T+2}, \dotsc, a_i^{m-T}$ is $\ell_i$. Since they cannot all be in the independent set, we have $\ell_i = 0$. Hence there are at most $2Tp + k$ nodes that output~$1$. By construction, $2Tp \le 2Tn/m \le n/(8\alpha)$ and $k = n - mp \le N \le n/(8\alpha)$.
\end{proof}

This immediately gives a negative result for the approximability of a maximum matching as well: given a matching $M$ in a numbered directed $n$-cycle, we can construct an independent set $I = \{ u : (u,v) \in M \}$ with $\mysize{I} = \mysize{M}$.
\begin{corollary}[\mythmcite{czygrinow08fast}]\label{cor:neg-m-general}
    There is no local constant-factor approximation algorithm for the maximum matching problem.
\end{corollary}

\section{Lower bounds and local reductions}\label{sec:lower}

In this section, we present local reductions that establish lower bounds for local approximation algorithms. All reductions are from the maximum independent set problem in numbered directed cycles (Theorem~\ref{thm:neg-is-cycle}). The reductions yield the strongest possible negative results, as there is a matching positive result for each of them. As an introduction to the local reductions, we begin with a known result for general graphs; Theorem~\ref{thm:neg-ds-general} is a restatement of the negative results for planar graphs~\cite{czygrinow08fast} and unit-disk graphs~\cite{lenzen08leveraging}.

\begin{theorem}[\mythmcite{czygrinow08fast,lenzen08leveraging}]\label{thm:neg-ds-general}
    For any even $\Delta \ge 2$ and $\eps > 0$, there is no local algorithm with approximation factor ${(\Delta+1-\eps)}$ for the minimum dominating set problem.
\end{theorem}
\begin{proof}
    Suppose that such an algorithm $\myA$ exists for some $\Delta=2k$. Let $\alpha=\Delta(\Delta+1)/\eps$. We will use $\myA$ to find an independent set with at least $n/\alpha$ nodes in numbered directed $n$-cycles for any $n$ divisible by $\Delta+1$. This is a contradiction with Theorem~\ref{thm:neg-is-cycle}.

    Given an $n$-cycle $\myC$ (Figure~\ref{fig:reductions}a), we construct the $2k$-regular graph $\myG = \myC^k$ (Figure~\ref{fig:reductions}b illustrates the case $k=2$); the node identifiers are inherited from the cycle $\myC$. We simulate the algorithm $\myA$ in the graph $\myG$. There is a dominating set of $\myG$ with $n/(\Delta+1)$ nodes; hence $\myA$ must output a dominating set $D$ with at most
\[
    \Bigl(1 - \frac{\eps}{\Delta+1}\Bigr) n
\]
nodes. Thus $\mysize{V \setminus D} \ge \eps n/(\Delta+1)$. The subgraph of $\myC$ induced by $V \setminus D$ consists of paths with at most $\Delta$ nodes each; hence there are at least
\[
    \frac{\eps n}{\Delta(\Delta+1)} = \frac{n}{\alpha}
\]
such paths. Construct an independent set $I$ with $\mysize{I} \ge n/\alpha$ by taking the first node of each such path.
\end{proof}

\begin{figure}
  \centering
  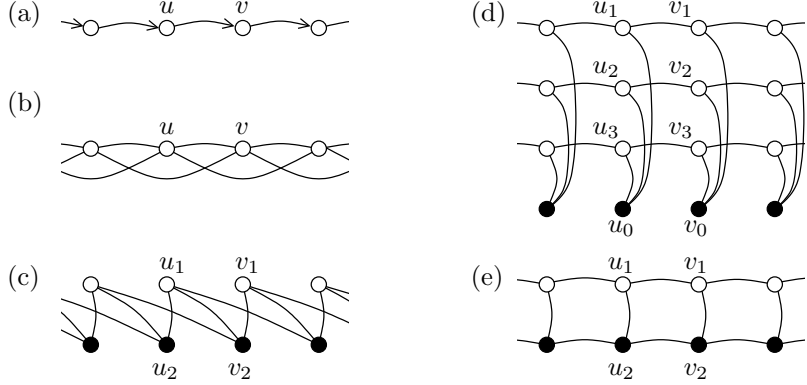
  \caption{The local reductions for the lower bounds.}\label{fig:reductions}
\end{figure}

\begin{theorem}\label{thm:neg-ds-strong}
    For any $\Delta \ge 2$ and $\eps > 0$, there is no local algorithm with approximation factor ${(\Delta+1)/2-\eps}$ for dominating sets in $2$-coloured graphs.
\end{theorem}
\begin{proof}
    Suppose that such an algorithm $\myA$ exists. Let $\alpha = (\Delta^2-1)/(2\eps)$. We will use $\myA$ to find an independent set with at least $n/\alpha$ nodes in numbered directed $n$-cycles for any $n$ divisible by $\Delta+1$. This is a contradiction with Theorem~\ref{thm:neg-is-cycle}.

    Given an $n$-cycle $\myC$ (Figure~\ref{fig:reductions}a), we construct a $\Delta$-regular $2$-coloured graph $\myG$ as follows (Figure~\ref{fig:reductions}c shows the case $\Delta=3$). For each node $v$ in $\myC$, there is a white node $v_1$ and a black node $v_2$ in $\myG$. If the directed path from $u$ to $v$ in $\myC$ has at most $\Delta-1$ edges, then there is an edge $\{u_1, v_2\}$ in $\myG$. The node identifiers are inherited from the cycle~$\myC$: for example, let $v_1 = 2v-1$ and $v_2 = 2v$.

    We simulate the algorithm $\myA$ in the graph $\myG$. There is a dominating set of $\myG$ with $2n/(\Delta+1)$ nodes; hence $\myA$ must output a dominating set $D$ with at most
\[
    \Bigl(1-\frac{2\eps}{\Delta+1}\Bigr) n
\]
nodes. Let $B = \{ v \in V : v_1 \notin D\textrm{ and } v_2 \notin D \}$; we have
\[
    \mysize{B} \ge \mysize{V} - \mysize{D} \ge \frac{2\eps n}{\Delta+1}.
\]
The subgraph of $\myC$ induced by $B$ consists of paths with at most $\Delta-1$ nodes each; hence there are at least
\[
    \frac{2\eps n}{(\Delta+1)(\Delta-1)} = \frac{n}{\alpha}
\]
such paths. Construct an independent set $I$ with $\mysize{I} \ge n/\alpha$ by taking the first node of each such path.
\end{proof}

\begin{theorem}\label{thm:neg-m-weak}
    For any $\Delta \ge 3$ and $\eps > 0$, there is no local algorithm with approximation factor ${(\Delta+1)/2-\eps}$ for maximum matching in weakly $2$-coloured graphs.
\end{theorem}
\begin{proof}
    Assume that such an algorithm $\myA$ exists. Let
    \[
        \eps' = \frac{2\eps}{\Delta+1-2\eps}, \qquad
        \alpha = \frac{2\Delta-1}{\eps'}.
    \]
    We will use $\myA$ to find an independent set with at least $n/\alpha$ nodes in numbered directed $n$-cycles for any even $n$. This is a contradiction with Theorem~\ref{thm:neg-is-cycle}.

    Given an $n$-cycle $\myC = (V_{\myC}, E_{\myC})$, we construct a weakly $2$-coloured graph $\myG$ as follows (Figure~\ref{fig:reductions}d shows the case $\Delta=3$). For each node $v$ in $\myC$, there are $\Delta+1$ nodes in $\myG$: white nodes $v_1, v_2, \dotsc, v_\Delta$ and a black node $v_0$. Each black node $v_0$ has degree $\Delta$: it is adjacent to all white nodes $v_1, v_2, \dotsc, v_\Delta$. Each white node has degree $3$: for each edge $(u,v)$ in $\myC$, there are edges $\{u_1, v_1\},\allowbreak \{u_2, v_2\}, \dotsc, \allowbreak \{u_\Delta, v_\Delta\}$ in $\myG$.

    There is a matching with $(\Delta+1)n/2$ edges in $\myG$. To see this, let $X$ be a perfect matching in $\myC$, with $\mysize{X} = n/2$. Construct a perfect matching in $\myG$ as follows: for each edge $(u,v)$ in $X$, choose the edges $\{v_0, v_\Delta\}$, $\{u_0, u_\Delta\}$, and $\{u_i, v_i\}$ for each $i \in \{1,2,\dotsc,\Delta-1\}$.

    We simulate the algorithm $\myA$ in the graph $\myG$. The algorithm must output a matching $M$ with at least ${(1+\eps')}n$ edges. Since there are $n$ black nodes in $\myG$, there are at least $\eps'n$ edges in $M$ that connect a pair of white nodes. For each $i = 1, 2, \dotsc, \Delta$, let
\[
    I_i = \bigl\{ u \in V_{\myC} : (u,v) \in E_{\myC},\, \{ u_i, v_i \} \in M \bigr\}.
\]
Now each $I_i$ is an independent set in $\myC$ and $\sum_i \mysize{I_i} \ge \eps' n$.

    We will now use the sets $I_i$ to construct an independent set $I$ in $\myC$ with $\mysize{I} \ge n/\alpha$. At least one of the sets $I_i$ satisfies this condition, but a local algorithm cannot find the right index $i$; hence we proceed as follows. We begin with $I = \emptyset$. At each iteration $i = 1, 2, \dotsc, \Delta$, for each node $v \in I_i$ in parallel, we (i)~add $v$ to $I$, and (ii)~remove the copy of $v$ and its neighbours from $I_i, I_{i+1}, \dotsc, I_\Delta$.

    In the end, $I$ is an independent set and each $I_i$ is empty. Furthermore, for each node added to $I$ there are at most $2\Delta-1$ nodes that we removed from $I_1, I_2, \dotsc, I_\Delta$; the worst case is that $I_1$ contains a node $v$ and each of $I_2, I_3, \dotsc, I_\Delta$ contains the two neighbours of $v$. Hence $\mysize{I} \ge \eps' n / {(2\Delta-1)} = n/\alpha$.
\end{proof}

\begin{theorem}\label{thm:neg-is-weak}
    There is no local constant-factor approximation algorithm for independent set in weakly $2$-coloured graphs.
\end{theorem}
\begin{proof}
    The reduction is from an $n$-cycle $\myC$ to the $3$-regular weakly $2$-coloured graph $\myG$ illustrated in Figure~\ref{fig:reductions}e. If we can find an independent set with at least $k$ nodes in $\myG$, then techniques similar to those in the proof of Theorem~\ref{thm:neg-m-weak} can be used to construct an independent set with at least $k/3$ nodes in $\myC$. By Theorem~\ref{thm:neg-is-cycle}, we must have $k = o(n)$.
\end{proof}

\section{Algorithms for weakly coloured graphs}\label{sec:weak-algo}

In this section we give a local algorithm to find a spanning forest of stars in a weakly $2$-coloured graph. Once the stars are formed, it is simple to find a dominating set (roots of the stars) and a matching (one edge for each star).

To find a small dominating set, we would prefer large (high-degree) stars, and to find a large matching, we would prefer small (low-degree) stars. Nevertheless, the same approach -- find \emph{any} set of stars -- yields the \emph{same} approximation factor ${(\Delta+1)}/2$ for both problems. Moreover, this is the best possible (Theorems \ref{thm:neg-ds-strong} and \ref{thm:neg-m-weak}).

To build the stars, we can use an algorithm that is similar to the \emph{Balanced\_DOM} subroutine in Kutten and Peleg~\cite{kutten98fast}. We first construct a forest $F$ of rooted trees (Figure~\ref{fig:weak-algo}); the directed edges in $F$ point towards the trees' roots. The construction is simple:
\begin{enumerate}
    \item Each black node $b$ chooses a white neighbour $w$; add the edge $(b,w)$ to~$F$ (Figure~\ref{fig:weak-algo}b).
    \item Each white node $w$ which does not have any children in $F$ chooses a black neighbour $b$; add the edge $(w,b)$ to~$F$ (Figure~\ref{fig:weak-algo}c).
\end{enumerate}

\begin{figure}
    \centering
    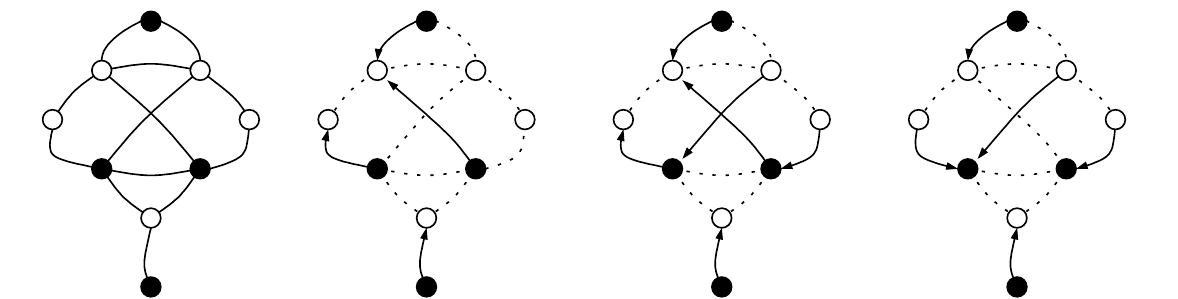
    \caption{Finding stars in a weakly $2$-coloured graph.}\label{fig:weak-algo}
\end{figure}

At this point, every node belongs to a tree; the depth of each tree is $1$ or $2$. Next we make local modifications within each tree, depending on its structure. Let $r$ be the root of the tree.
\begin{enumerate}
    \item If all leaves are at depth $1$, do nothing.
    \item If there are leaf nodes both at depth $1$ and at depth $2$, remove all edges $(c,r)$ where $c$ is a non-leaf child.
    \item Otherwise, choose arbitrarily a child $x$ of the root. Remove all edges $(c,r)$ where $c$ is a child of the root, and $c \neq x$. Reverse the edge $(x,r)$.
\end{enumerate}
Now $F$ consists of stars, i.e., rooted trees of depth $1$ (Figure~\ref{fig:weak-algo}d). Each node is either a root node with at least one child, or a leaf node. The algorithm can be implemented by using only a port numbering; unique node identifiers are not needed. The port numbers are used both for representing the forest $F$ (e.g., a child does not know the identity of the parent node, but it knows the port number of the edge that leads to the parent node) and for breaking ties (e.g., when a black node has to choose one of its white neighbours).

Next we present the applications of the stars.
\begin{theorem}\label{thm:pos-ds-weak}
    For any $\Delta \ge 1$, there is a local algorithm with approximation factor ${(\Delta+1)/2}$ for dominating set in weakly $2$-coloured graphs.
\end{theorem}
\begin{proof}
    Let $D$ be the set of the roots of the stars. The set $D$ is a dominating set with at most $\mysize{V}/2$ nodes. Let $D^{*}$ be a minimum dominating set. Since a node cannot dominate more than $\Delta$ neighbours, $\mysize{D^{*}} \ge \mysize{V}/{(\Delta+1)}$. Therefore $D$ is a ${(\Delta+1)/2}$-approximation of a minimum dominating set.
\end{proof}

\begin{theorem}\label{thm:pos-m-weak}
    For any $\Delta \ge 1$, there is a local algorithm with approximation factor ${(\Delta+1)/2}$ for maximum matching in weakly $2$-coloured graphs.
\end{theorem}
\begin{proof}
    Let $M$ be the set of edges with one edge chosen arbitrarily from each star. The set $M$ is a matching. Each star contains at most $\Delta+1$ nodes; hence $\mysize{M} \ge \mysize{V}/{(\Delta+1)}$. For an optimal matching $M^{*}$, we have $\mysize{M^{*}} \le \mysize{V}/2$. Hence $M$ is a ${(\Delta+1)/2}$-approximation of a maximum matching.
\end{proof}

\section{Approximating dominating set if \texorpdfstring{$\Delta$}{Delta} is odd}

Now we are ready to present an application of the Naor--Stockmeyer algorithm for weak $2$-colouring~\cite{naor95what} and the techniques that we developed in Section~\ref{sec:weak-algo}. In this section we assume that the graph is not only port-numbered, but there is also an orientation: for each edge $\{u,v\} \in E$, exactly one direction $(u,v)$ or $(v,u)$ has been chosen.

The orientation can be used to break the symmetry in some cases. Specifically, if each node of $\myG$ has an odd degree, then we can use Naor and Stockmeyer's algorithm to find a weak $2$-colouring; the algorithm does not require unique identifiers~\cite{mayer95local}. Theorem~\ref{thm:pos-ds-weak} then provides a factor ${(\Delta+1)/2}$ approximation for dominating set.

However, in this section we study the case where the degree bound $\Delta$ is odd, but nothing else is known about the degrees of the graph; that is, the case for which we have the lower bound $\Delta-\eps$ from Theorem~\ref{thm:neg-ds-general}. A combination of weak colouring and Theorem~\ref{thm:pos-ds-weak} provides a matching upper bound.

\begin{theorem}\label{thm:pos-ds-odd}
    For any odd $\Delta \ge 1$, there is a local algorithm with approximation factor $\Delta$ for dominating set in graphs with maximum degree~$\Delta$, assuming that there is a port numbering and an orientation.
\end{theorem}
\begin{proof}
Partition $V$ into $V = A \cup B \cup C$ such that $A$ consists of the odd-degree nodes, $B$ consists of the even-degree nodes adjacent to at least one node in $A$, and $C$ is the rest; in particular, the degree of each node in $B$ or $C$ is at most $\Delta-1$.

Consider the subgraph $\myH$ induced by $A\cup B$. In the subgraph $\myH$, the degree of each node in $A$ is odd, but some of the nodes in $B$ may have an even degree. Construct a new graph $\myH_2$ by adding a new dummy node of degree $1$ as a neighbour of each even-degree node in the subgraph $\myH$. Now every node in $\myH_2$ has an odd degree, and we can use the Naor--Stockmeyer algorithm~\cite{mayer95local} to weakly $2$-colour it.

At this point, each node in $A$ is adjacent to a node of the opposite colour, but this does not necessarily hold for the nodes in $B$. However, we can easily find valid colours for each node in $B$ in parallel: if $b \in B$ and each $a \in A$ adjacent to $b$ has the same colour as $b$, then we reverse the colour of $b$. Now each node in $B$ has a neighbour with the opposite colour in $A$; furthermore, no node in $A$ lost a neighbour of the opposite colour.

Thus $\myH$ is weakly $2$-coloured, and we can apply the algorithm of Theorem~\ref{thm:pos-ds-weak} to find a dominating set $D_\myH$ with $\mysize{D_\myH} \le \mysize{A \cup B}/2$ in the subgraph $\myH$. The set $D = D_\myH \cup C$ is now a dominating set of the original graph $\myG$.

Let $D^*$ be a minimum dominating set of $V$. Let $D_1^* = D^* \cap A$ and $D_2^* = D^* \cap (B \cup C)$. Since a node with a degree $d$ can dominate at most $d+1$ nodes, and the nodes in $D_1^*$ are not adjacent to the nodes in $C$, the set $D^*$ must satisfy
\begin{align*}
    (\Delta+1)\mysize{D_1^*}+\Delta\mysize{D_2^*} &\ge \mysize{A}+\mysize{B}+\mysize{C}, \\
    \Delta \mysize{D_2^*} &\ge \mysize{C},
\intertext{which implies}
    \mysize{D^*} = \mysize{D_1^*}+\mysize{D_2^*} &\ge \frac{\mysize{A}+\mysize{B}}{\Delta+1}+\frac{\mysize{C}}{\Delta}.
\end{align*}
Since $\mysize{D} \le (\mysize{A}+\mysize{B})/2+\mysize{C}$, we have $\mysize{D}/\mysize{D^*} \le \Delta$.
\end{proof}

\begin{remark}
    It is necessary to assume that the graph is oriented in Theorem~\ref{thm:pos-ds-odd}. If the nodes are anonymous and there is a port numbering but no orientation, a deterministic distributed algorithm cannot have a better approximation factor than $\Delta+1$. To see this, consider the complete graph $K_{\Delta+1}$ on $\Delta+1$ nodes. Find an edge colouring of $K_{\Delta+1}$ with $\Delta$ colours -- this is possible, since we assumed that $\Delta$ is odd. Use the edge colouring to assign the port numbers: an edge with colour $k$ has the port number $k$ in both ends. Now from the perspective of distributed algorithms, the nodes are indistinguishable. Any deterministic algorithm has to produce the same output for each node; in particular, it has to output a dominating set with $\Delta+1$ nodes, while $1$ node would suffice.
\end{remark}

\section{Matching in two-coloured graphs}\label{sec:mm-scheme}

In Sections~\ref{sec:lower} and~\ref{sec:weak-algo} we proved that in weakly $2$-coloured graphs the maximum matching problem can be approximated to within a factor of $(\Delta + 1) / 2$, but not better. In this section we show that in $2$-coloured graphs the problem has a local approximation scheme.

Given a matching $M$, an augmenting path (w.r.t.\ $M$) is a path that starts and ends at an unmatched node and whose every other edge belongs to $M$. An augmenting tree is a tree whose every root--leaf path is an augmenting path. In Figures~\ref{fig:mm-scheme}a--c, a matching, two augmenting trees (rooted at black nodes), and two augmenting paths are shown.

\begin{figure}
    \centering
    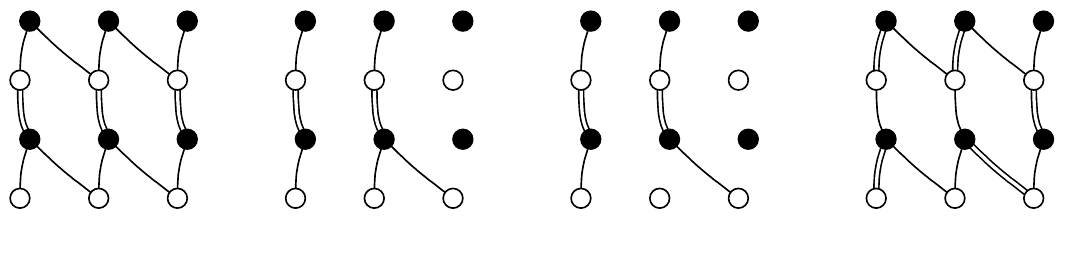
    \caption{Finding length-$3$ augmenting paths in $2$-coloured graphs with a local algorithm. (a)~The graph~$\myG$. A matching~$M$ is highlighted with double lines. The matching is maximal, i.e., there is no length-$1$ augmenting path. However, there are several length-$3$ augmenting paths. (b)~Augmenting trees. The set of root nodes is a \emph{subset} of the black endpoints of the length-$3$ augmenting paths, while the set of leaf nodes is \emph{equal} to the set of white endpoints of the length-$3$ augmenting paths. (c)~Augmenting paths, one per tree. (d)~The new matching.}\label{fig:mm-scheme}
\end{figure}

The symmetric difference of $M$ and an augmenting path is a new matching whose size is larger than the size of $M$ by~1. If every augmenting path is longer than $2k - 1$, then the size of $M$ is at least $k/{(k+1)}$ times the size of the maximum matching (folklore). Hence we have the following $(1 + 1/k)$-approximation algorithm for maximum matching: Starting from an empty matching, for each $i = 1, 2, \dotsc, k$, find repeatedly augmenting paths of length $2i - 1$ and augment along the paths, until no such path exists. The iteration $i=1$ is equal to finding a maximal matching, which could be done locally by Ha\'{n}\'{c}kowiak \myetal{}'s~\cite{hanckowiak98distributed} algorithm. Below we give a local algorithm that implements the iteration $i$ for a general $i \ge 1$. The algorithm uses techniques presented by Balas \myetal{}~\cite{balas91parallel} in the context of parallel algorithms.

Our algorithm repeatedly invokes a subroutine that removes \emph{some} augmenting paths of length $h=2i - 1$, assuming that there is no shorter augmenting path. The subroutine consists of three phases (refer to Figure~\ref{fig:mm-scheme}).
\begin{enumerate}
    \item In the \emph{flooding phase}, we construct a forest $F$ of disjoint augmenting trees, rooted at black nodes (Figure~\ref{fig:mm-scheme}b). Each root--leaf path has length $h$. Furthermore, if there is a length-$h$ augmenting path in the original graph between a black node $b$ and a white node $w$, then $w$ is a leaf node in a tree of $F$. (However, $b$ may or may not be a root node in a tree of $F$.)
     \item In the \emph{proposal phase}, we choose one augmenting path in each tree (Figure~\ref{fig:mm-scheme}c).
    \item Finally, in the \emph{augmenting phase}, we augment along the paths in parallel to find a new matching (Figure~\ref{fig:mm-scheme}d).
\end{enumerate}

To implement the flooding phase, each unmatched black node considers itself as a potential root of an augmenting tree. Every root node sends a message to each of its neighbours. When a white node receives messages, it chooses one of the senders as its parent, and forwards the message to its neighbour along an edge in the matching. When a black node receives a message, it chooses the sender as its parent, and forwards the message to its neighbours along each edge that is not in the matching. Messages are propagated for $h$ hops; messages that reach a matched white node or a dead end are simply discarded.

Since there is no augmenting path shorter than $h$, every unmatched white node that receives a message is an endpoint of a length-$h$ augmenting path. Conversely, all white endpoints of length-$h$ augmenting paths are reached by the messages. These unmatched white nodes become the leaves of the forest $F$. The edges of $F$ are defined by the links that point towards the parent nodes.

We now show that the trees of the forest $F$ are disjoint. To reach a contradiction, assume that $T_1$ and $T_2$ are two trees in $F$ and they share a node $v$. Let $b_j$ be the root node of the tree $T_j$; by assumption, $b_1 \ne b_2$. Let $P_j$ be an augmenting path in $T_j$ that begins from $b_j$, passes through $v$, and ends at a leaf node; let $\ell_j$ be the distance between $b_j$ and $v$ along $P_j$. If we had $\ell_1 = \ell_2$, the message initiated by the root $b_1$ would have reached the node $v$ on the same time step as the message initiated by the root $b_2$, and in our algorithm $v$ (or one of its ancestors) would have discarded one of the messages and joined only one of the trees. Hence we must have $\ell_1 \ne \ell_2$; but then it is possible to find an augmenting path (in the union of $P_1$ and $P_2$) that is strictly shorter than $h$, which contradicts our assumption.

Hence a local algorithm can find the forest $F$ with the above-mentioned properties. The other steps of the algorithm are straightforward. In the proposal phase, messages are initiated by the leaf nodes and propagated towards the root nodes; whenever several messages meet, all but one of them are discarded. Eventually, we have chosen exactly one augmenting path in each tree. Finally, in the augmenting phase, we augment along each of these paths in parallel.

To analyse how many invocations of the subroutine are needed, note that a white node can be an endpoint of at most $t_i = \Delta(\Delta-1)^{i-1}$ length-$h$ augmenting paths. Every invocation matches the other endpoint of at least one such path. Furthermore, it can be shown that no new augmenting paths with at most $h$ edges are created. Therefore, after $t_i$ invocations, there is no augmenting path with $h$ edges or fewer.

\begin{theorem}\label{thm:pos-m-strong}
    For any $\Delta \ge 1$ and $\eps > 0$, there is a local algorithm with approximation factor ${1+\eps}$ for maximum matching in $2$-coloured graphs. \qed
\end{theorem}

\section*{Acknowledgements}

This work was supported in part by the Academy of Finland, Grants 116547, 118653 (ALGODAN), and 132380, by Helsinki Graduate School in Computer Science and Engineering (Hecse), and by the Foundation of Nokia Corporation.

\def\bibfont{\small}

\end{document}